\documentclass[12pt]{amsart}
\usepackage[english]{babel}
\usepackage{latexsym,amsmath,amsfonts,amscd,amssymb,latexsym}
\usepackage{graphicx}
\usepackage{xcolor}
\usepackage{tikz}
\usetikzlibrary{automata}
\usetikzlibrary{positioning}  
\usetikzlibrary{arrows}       
\tikzset{node distance=4.5cm, 
         every state/.style={ 
           semithick,
           fill=gray!10},
         initial text={},     
         double distance=4pt, 
         every edge/.style={  
         draw,
           ->,>=stealth',     
           auto,
           semithick}}
\newtheorem{theorem}{Theorem}[section]
\newtheorem{theorem*}{Theorem}

\newtheorem{proposition}[theorem]{Proposition}
\newtheorem{proposition*}[theorem*]{Proposition}
\newtheorem{corollary}[theorem]{Corollary}
\newtheorem{corollary*}[theorem*]{Corollary}
\newtheorem{definition}[theorem]{Definition}

\newtheorem{example}[theorem]{Example}

\theoremstyle{remark}
\newtheorem{remark}[theorem]{Remark}
\newtheorem{remark*}[theorem*]{Remark}
\newtheorem{note}[theorem]{Note}
\newtheorem{note*}[theorem*]{Note}

\newcommand{\Off}{\operatorname{Off}}
\newcommand{\Orph}{\operatorname{Orph}}

\newcommand{\EE}{{\mathbb E}}

\def\bitcoinA{%
  \leavevmode
  \vtop{\offinterlineskip 
    \setbox0=\hbox{B}%
    \setbox2=\hbox to\wd0{\hfil\hskip-.03em
    \vrule height .3ex width .15ex\hskip .08em
    \vrule height .3ex width .15ex\hfil}
    \vbox{\copy2\box0}\box2}}

\begin{document}

\title{Block withholding resilience}


\author{Cyril Grunspan}
\address{L{\'e}onard de Vinci, P{\^o}le Univ., Research Center, Paris-La D{\'e}fense, France}
\email{cyril.grunspan@devinci.fr}
\author{Ricardo P\'erez-Marco}
\address{CNRS, IMJ-PRG,  Paris, France}
\email{ricardo.perez.marco@gmail.com}

\begin{abstract} 
 It has been known for some time that the Nakamoto consensus as implemented in the Bitcoin protocol 
 is not totally aligned with the individual interests of
  the participants. More precisely, it has been shown that block withholding mining strategies 
  can exploit the difficulty adjustment algorithm of the protocol and obtain an unfair advantage.
  However, we show that a modification of the difficulty adjustment
  formula taking into account orphan blocks makes honest mining the only optimal strategy.  
  Surprinsingly, this is still true when orphan blocks are rewarded with an amount smaller to the official 
  block reward. This gives an incentive to signal orphan blocks.
  The results are independent of the connectivity of the attacker.
\end{abstract}

{\maketitle}

\section{Introduction}

Satoshi Nakamoto's foundational article describes the Bitcoin protocol \cite{N08}.  
Bitcoin is an electronic currency and bitcoin transactions operate through a computer network. 
This network is \textit{permissionless}: anyone can freely enter or leave the 
network \cite{Warren2023}. Moreover, there is no central authority acting as a referee. The minting algorithm is implemented in the protocol. 
Transactions are packed in chronologically ordered blocks that form an unforgeable public ledger: \textit{the blockchain}. 
Certain nodes of the network, called miners, play a special role. They secure the blockchain through intensive 
computation by a ``proof of work'', 
a technique originally invented to fight email spam and denial of service attacks. A miner 
creates a new block of transactions to add to the blockchain 
by solving a cryptographic puzzle by brute force iterating a simple algorithm. For this computational work he is rewarded
by a \textit{coinbase reward} of newly minted bitcoins. This mechanism generates the Bitcoin monetary mass. 

This article investigates mining strategies for Bitcoin. A mining strategy determines each time 
where the miner should mine according to the history of the discovered blocks. Most of the time 
he has the choice between mining on the last block of the official blockchain or on top of a fork that he keeps secret. 
The strategy also specifies whether the miner should make public any previously discovered blocks that he has kept secret.
Satoshi Nakamoto proposes mining always on top of the last block of the official blockchain, and always making public his discovered blocks. 
He thought that this strategy called ``honest strategy" was the most profitable strategy. But, indeed, other more profitable strategies 
are possible.

Mining is at the heart of Nakamoto's consensus. Nakamoto's consensus has the remarkable feature that it is a probabilistic 
valid consensus in an open network, whereas former known consensus only existed in closed systems. In general, 
when the miner sticks to a mining strategy, the evolution of the network from different states follows the random discovery of blocks. 
The number of such states is countable. A mining strategy gives a dynamical system in its associated Markov chain. The strategy is 
repetitive: regularly, and in finite time, the miner comes back to the initial state, i.e. to mine on the last block of the official blockchain 
which corresponds to the state $0$. The Markov chain is therefore recurrent \cite{DMPS20}. 
This ``state machine'' model is popular in the literature. Each transition on the Markov chain comes with rewards 
earned by the miners, but these rewards are not immediately earned by the miners but later 
(when the network returns to the initial state). Then it is possible to calculate the profitability of a  
mining strategy. On the long run, when the difficulty adjustment has stabilized, official blocks arrive on average every ten minutes, 
then the profitability is proportional to the percentage of blocks 
mined by the miner in the official blockchain (see Section \ref{deltaT} below). One can compute this proportion by computing 
the stationary probability of the Markov chain, which only requires to invert the 
transition matrix \cite{ES14}. If one wishes to prove general optimality results that hold for any mining strategy, such as the two 
results presented in this article, we need to proceed differently. The main technique is the use of Martingale Theory. 
For this, we extend the point of view of Satoshi Nakamoto 
when he tried to model the evolution of the network when computing the probability of success of a double spend attack 
using Poisson processes \cite{N08, GPM17, GZ19}. In this setup we are able to use Doob's martingale theorem. It is difficult to 
conceive a proof of these results without using Martingale Theory and the authors do not know of an alternative proof. 

The two general results presented here (Theorem \ref{1st} and Theorem \ref{2nd}) are a culminate part of the literature on deviant mining 
strategies.
These strategies have in common to be block withholding mining where the miner does not always mine on the last block of 
the official blockchain and times properly when to make his block findings public. 
After having showed the existence of such strategies and calculating their profitability, we looked for the optimal strategy 
highlighting a counter-intuitive phenomenon: in some cases, the miner may have an interest in mining against the whole 
network on a secret fork whose height is lower than the height of the official blockchain (trailing mining) \cite{B13, ES14, N16, SSZ17}. 
Several solutions were then put forward to limit the effectiveness of deviant strategies. Then, it was understood that this was an attack 
on a weakness of the Difficulty Adjustment Algorithm (DAA) \cite{GPM2018}. Some qualitative arguments were presented on how to change the DAA 
taking into account the production of orphan blocks. Ethereum implemented this partially rewarding uncle blocks, 
before adopting a completely different consensus mechanism in 2022 (Proof of Stake). The literature on 
the study of deviant mining strategies has been abundant and the reader can consult \cite{L22, Warren2023} for a survey on the subject. 
Lately, authors have revisited Bitcoin network mining, identifying the miner's ruin problem as a dual of a classical insurance problem. Then they considered the possibility of ruin when a miner uses a deviant strategy \cite{HG22}.

Several recent papers with numerical simulations have confirmed the relevance of including the orphan block count in the 
DAA \cite{Zhang23, Zhou23}. The main result of this article is to show that a modification of Bitcoin's DAA including the count of orphan blocks
neutralizes block withholding attacks because they become non-profitable (Theorem \ref{2nd}). It is proved that we can even reward all orphan blocks 
and this ensures that they are recorded in the blockchain data.
It is conceivable that the novel findings presented herein (Theorem \ref{1st} and Theorem \ref{2nd}) may have been deliberated upon or even proved elsewhere, yet we remain uninformed of any citations pertaining to this matter.

This article is self-contained. We introduce a very simple mining strategy (the ``1+2 strategy'' in Section \ref{bwa}) 
which shows in the simplest way that the honest strategy is not optimal under the usual DAA. We start in the next section reviewing 
the basics of mining on Bitcoin \cite{GPM2020, Warren2023}, that can be skipped by the knowledgeable reader. 

\begin{note}
The term ``connectivity" as employed in the conclusion of the abstract, pertains to a parameter initially introduced by \cite{ES14}. This parameter, widely adopted by various authors and typically denoted as $\gamma$, is formally defined as a conditional probability. It quantifies the proportion of honest miners who choose to mine on top of an attacker's block in scenarios where two blocks of equal height are in competition, with one being mined by the attacker. Stating that 
``the results are independent of the connectivity of the attacker"
signifies that the results remain valid regardless of the specific value assigned to $\gamma.$
\end{note}

\section{How Bitcoin works}

 \subsection{Nodes}

The Bitcoin network is a peer-to-peer network made up of thousands of ``nodes"
connected around the world, forming an irreducible and highly connected graph. 
All nodes play or can play the same role and perform
the same operations. The nodes exchange information that consists of more or
less complicated transactions and blocks of validated transactions.
The blockchain is a database that defines the ledger of all confirmed
transactions. 
The mempool is the set of all transactions in waiting for a confirmation in a block.
Each time a node receives a transaction, it examines whether the transaction
is legal, in the proper format, and checks
that it does not conflict with other transactions already in its
mempool. When this is the case, the transaction is added to its mempool and broadcasted to
neighboring nodes.

\subsection{Mining}

Some of the nodes of the network perform a
mining activity. A miner is a particular node that seeks to build up a new block.
By definition, a \ensuremath{\mathcal{B}} block is a set of data whose maximum size is
about 2 Mega Bytes. It is formed by a reference to an old block,
a set of transactions drawn from the miner's mempool, a creation date, the
mining difficulty $\Delta$ and a parameter called "nonce". The goal is to
obtain a block \ensuremath{\mathcal{B}} such the hash of its header is below some threshold:
\begin{equation}\label{1Delta}
  f (\mathcal{B}) < \frac{1}{\Delta}
\end{equation}
\label{hash}where $f$ is the ${\text {SHA}}\_256\,{\circ}\,{\text {SHA}}\_256/(2^{256}-1)$ function where SHA\_256 is
the cryptographical hash function, and $\Delta$ is called the \textit{difficulty}. If successful, the
miner receives a reward as a coinbase transaction. The reward decreases by half every $210\,000$ blocks, approximately every four years (the ``halving''). The difficulty parameter $\Delta$ is adjusted regularly, in periods of 2016 blocks\footnote{A bug in the Bitcoins code makes that 2015 was implemented instead of 2016. We will ignore this fact that is irrelevant for our analysis.}, that is in about $2$ weeks in average. The current 
difficulty adjustment algorithm 
allows blockwithholding attacks, i.e. the miner can obtain an advantage by withholding blocks 
and releasing them at an appropriate time. This type of attack exploits the difficulty adjustment algorithm that 
evaluates the total hashrate in an erroneous way, only considering validated blocks and forgetting about orphan 
blocks (that fulfill the proof of work but are displaced from the blockchain by other competing blocks). 
The purpose of this article
is to show that this problem can be corrected using an improved
difficulty adjustment mechanism that evaluates properly the total hashrate by taking into account orphan blocks.

\subsection{Mining strategies}

Why should miners immediately release a newly discovered block? This
is implicitely assumed in the Bitcoin founding article, but in fact it is not properly incentivized.
Could a miner validating a block try to widen the gap even
further by secretly mining on top of that block? Could he devise a block withholding strategy 
giving him an advantage compared to the honest strategy? 

\begin{definition}
  The official blockchain is the chain of blocks with the most of proof-of-work.
\end{definition}
In other words, the official blockchain is the chain of blocks which maximizes 
the quantity $\sum \Delta_i$ where $\Delta_i$ is the difficulty parameter included in the $i$-th block. 
Most ot the time it is the longest blockchain because the difficulty parameter 
is locally constant.

\begin{definition}
  The honest strategy consists in always mining on top of the last
  block of the official blockchain.
\end{definition}

\subsection{Performance of a mining strategy}

A  mining strategy is repetitive in the sense that the miner will return
to its starting point after a finite time. The miner has performed a strategy cycle during this period. 
We call such a strategy with finite expected return time, a \textit{finite strategy} 
(also \textit{integrable strategy} has been used). We note $\tau$ the
random duration of this cycle and $G$ the miner's gain accumulated during
this cycle. We use as a unit of wealth the average value
contained in a block (i.e., average value of transaction fees contained in a
block plus value of a coinbase). Mathematically speaking $\tau$ is a
stopping time. We consider only cycles with expected finite duration:
$\mathbb{E} [\tau] < \infty$. If a miner repeats his strategy $n$ times, 
he will gain per unit time
\begin{eqnarray*}
  \frac{G_1 + \ldots + G_n}{\tau_1 + \ldots + \tau_n} & = & \frac{\frac{G_1 +
  \ldots + G_n}{n}}{\frac{\tau_1 + \ldots + \tau_n}{n}}
\end{eqnarray*}
which converges to $\frac{\mathbb{E} [G]}{\mathbb{E} [\tau]}$ according to the strong law of large numbers (the assumption 
$\mathbb{E} [\tau] < \infty$ implies that $\mathbb{E} [G]
< \infty$). Similarly, the cost per unit time of the miner over the long term
is $\frac{\mathbb{E} [C]}{\mathbb{E} [\tau]}$ where $C$ is the cost per cycle
of his mining activity and the miner's net income per unit time is therefore
$\frac{\mathbb{E} [G]}{\mathbb{E} [\tau]} - \frac{\mathbb{E} [C]}{\mathbb{E}
[\tau]}$. \ However, the cost of his mining activity per unit of time does not
depend on whether or not the miner keeps blocks secret (it
depends on the cost of electricity, the price of his equipment, the salaries
paid to employees for a mining company, etc). Whether he makes blocks public or not has no impact
on his cost per unit of time. Therefore, comparing two mining strategies with
the same average operating cost per unit of time, a rational miner will choose
the strategy that maximizes his income per unit of time in the long run
$\Gamma = \frac{\mathbb{E} [G]}{\mathbb{E} [\tau]}$.

\begin{definition}
  The profitability ratio of an integrable mining activity is $\Gamma=\frac{\mathbb{E}
  [G]}{\mathbb{E} [\tau]}$ where $G$ is the gain per cycle and $\tau$ the
  duration of a cycle.
\end{definition}

\begin{note}\label{unitime}
According to the profitability analysis from \cite{GPM2018} (see also \cite{BET20, SSZ17}), 
what counts is the yield per 
unit time in the long run, that is, the limit of
$\frac{G_1 + \ldots + G_n}{\tau_1 + \ldots + \tau_n}$ which converges to 
$\frac{\mathbb{E} [G]}{\mathbb{E} [\tau]}$, and not the average yield per unit time per 
cycle $\EE\left[\frac{G}{\tau}\right]$. The two quantities are not equal. 
See also Note \ref{objective} below.
\end{note}

\section{Mining model and notations}

We use the widely adopted model and notations that Satoshi Nakamoto proposed 
in his Bitcoin founding article (see \cite{N08} and \cite{GPM17}).
We consider a miner who is a possible attacker (or a group of miners)
against the rest of the network consisting of honest miners all
following the rules of the Bitcoin protocol. 
Regularly, miners (attacker or honest miners) add blocks chained to each other on a rooted graph whose root is the first Bitcoin block called the genesis block. For any real positive $t$, the set of mined blocks up to time $t$ forms a history which defines a filtration 
$(\mathcal{F}_t)$.

\begin{definition}
  For $t \in \mathbb{R}_+$, we denote $N (t)$ (resp. $N' (t)$) the number of
  blocks mined by honest miners (resp. attacker) between $0$ and $t$.
\end{definition}

Since the hash
function used in  Bitcoin mining operations is pseudorandom, the time
$\mathbf{T}$ (resp. $\mathbf{T}'$) taken by the honest miners (resp.
attacker) to discover a block is a random variable following an exponential law (see \cite{GPM17}). We denote by
$\alpha$ (resp. $\alpha'$) the parameter of the exponential law. We have $\mathbb{E} [\mathbf{T}]
= \frac{1}{\alpha}$ (resp. $\mathbb{E} [\mathbf{T}'] = \frac{1}{\alpha'}$). In
other words, $\alpha$ (resp. $\alpha'$) is the average speed taken by honest
miners (resp. attacker) to discover a block.

The counting processes $N (t)$ and $N' (t)$ are Poisson processes of
parameters $\alpha$ and $\alpha'$ with respect to filtration $(\mathcal{F}_t)$ \cite{GPM17}. Let $h$ (resp. $h'$) be the number of
hashes per second computed by the honest miners (resp. attacker). This is the
absolute hash power of the honest miners (resp. attacker). 

Assuming the SHA\_256 function is perfect, the distribution of the images of the values of the function $f={\text {SHA}}\_256\,{\circ}\,{\text {SHA}}\_256/(2^{256}-1)$ is uniformly distributed in   
$[0,1]$ in first approximation. On average, therefore, $\Delta$ iterates of $f$ must be calculated before finding a value whose image belongs to the interval $[0,1/\Delta]$. Consequently, according to (\ref{1Delta}), the difficulty parameter $\Delta$ is also the average number of attempts required before obtaining a proof of work. A miner who is able to compute $h$ hashes per second will therefore have to wait an average time 
$\EE[\mathbf{T}]$ before obtaining a proof of work with $h\cdot\EE[\mathbf{T}] = \Delta$. In the  same way, we have:
$h' \cdot \mathbb{E} [\mathbf{T}'] = \Delta$.

It follows that
\begin{eqnarray}
  \alpha & = & \frac{h}{\Delta}  \label{alphaDel}\\
  \alpha' & = & \frac{h'}{\Delta}  \label{alphaprimeDel}
\end{eqnarray}
Since the two variables $\mathbf{T}$ and $\mathbf{T}'$ are independent,
we also have $\mathbb{P} [\mathbf{T}' <\mathbf{T}] = \frac{\alpha'}{\alpha +
\alpha'} = \frac{h'}{h + h'}$. Therefore, the probability that the attacker
discovers a block before the other miners is equal to its relative hash power.
We make the assumption that this remains constant over time.
We keep Nakamoto's notation for relative hashing powers of miners \cite{N08}.

\begin{definition}
  We denote by $p$ (resp. $q$) be the relative hashing power of the honest miners
  (resp. attacker).
\end{definition}

In other words, $q = \frac{\alpha'}{\alpha + \alpha'}$ and $p = 1 - q$. Note
that the network (the honest miners and the attacker) finds a block in $\inf
(\mathbf{T}, \mathbf{T}')$. This is again an exponential law of
parameter $\alpha + \alpha'$ because $\mathbf{T}$ and $\mathbf{T}'$ are
independent \cite{Ross}. Therefore, $\mathbb{E} [\inf (\mathbf{T},
\mathbf{T}')] = \frac{1}{\alpha + \alpha'}$ and $\alpha =
\frac{p}{\mathbb{E} [\inf (\mathbf{T}, \mathbf{T}')]}$ and $\alpha' =
\frac{q}{\mathbb{E} [\inf (\mathbf{T}, \mathbf{T}')]}$.
 
\subsection{First stability theorem}

In this section only, we consider a simplified Bitcoin network
without difficulty adjustment (that is, the difficulty parameter is assumed to be
constant). Let $\tau_0 =\mathbb{E} [\inf (\mathbf{T}, \mathbf{T}')]$
be the average time taken by the network to discover a block. Note that a
priori here, $\tau_0 \not= 10$ minutes because there is no particular
difficulty adjustment. It follows easily from the previous section that $N$ (resp. $N'$) is a Poisson process 
with parameter $\alpha = \frac{p}{\tau_0}$ (resp. $\alpha' =
\frac{q}{\tau_0}$) \cite{Ross}.

\begin{theorem}\label{1st}
  If $\Gamma$ is the profitability ratio for a finite mining strategy, we have
  $\Gamma \leqslant \frac{q}{\tau_0}$.
\end{theorem}

\begin{proof}
  Since the mining strategy is finite we have by definition $\mathbb{E} [\tau] < \infty$
  where $\tau$ the stopping time of the duration a cycle. Let $G$ be the number of
  official blocks mined by the attacker between 0 and $\tau$. We have $G \leqslant N'
  (\tau)$. Let $n \geq 1$ and the truncated stopping time $\tau_n = \inf
  (\tau, n)$. The random variable $N' (t) - \alpha' t$ is a martingale.
  Therefore, by Doob's Stopping Time Theorem, we have $\mathbb{E} [N' (\tau_n) - \alpha'
  \tau_n] = 0$, and $\mathbb{E} [N' (\tau_n)] = \alpha'
  \mathbb{E} [\tau_n]$. By making $n\to +\infty$ and using the Monotone
  Convergence Theorem, we obtain $\mathbb{E} [N' (\tau)] = \alpha'
  \mathbb{E} [\tau]$, which finally gives $\Gamma = \frac{\mathbb{E}
  [G]}{\mathbb{E} [\tau]} \leqslant \alpha' = \frac{q}{\tau_0}$.
\end{proof}

\begin{corollary}\label{cor:optimal}
  Without difficulty adjustment, the optimal strategy is the honest one.
\end{corollary}

\begin{proof}
  If the attacker mines honestly, he earns on average $q$ during a lapse $\tau_0$.
  Thus, the return on the honest strategy is \ $\frac{q}{\tau_0}$ which is
  also the maximum return in the previous Theorem. 
\end{proof}

This result was already proved in  \cite{GPM2018}.

\subsection{Effect of the difficulty adjustment on the profit per unit time}
\label{deltaT}

The difficulty adjustment ensures a lower bound for expected interblock 
time that allows the network to synchronize, and also an upper bound for transaction confirmation 
waiting time (i.e. inclusion in the blockchain). The goal of a regular difficulty adjustment
is to have, on average, a lapse of 10 minutes to validate a block. More precisely, the difficulty
parameter $\Delta$ is adjusted every $n_0$ blocks with $n_0 = 2016$. At each
difficulty adjustment, the network calculates the time $T$ taken to validate
the last series of $n_0$ blocks. This work is done thanks to the timestamps of
the blocks in the blockchain. If this time $T$ is greater than 14
days $= n_0 \times 10$ minutes, the difficulty decreases. Otherwise, it
increases. The new difficulty parameter $\Delta'$ is given by:
\begin{equation}
  \Delta' = \Delta \times \frac{n_0 \times 10}{T} \label{cible} 
\end{equation}
where $T$ is calculated here in minutes. When the mining difficulty $\Delta$
is changed, so are the parameters $\alpha$ and $\alpha'$ of the Poisson
processes $N$ and $N'$. The counting processes $N$ and $N'$ are piecewise
Poisson processes whose parameters $\alpha$ and $\alpha'$ change each time the
official blockchain progresses by $n_0$ blocks. If $\Delta$ varies and becomes
equal to $\frac{\Delta}{\lambda}$, then $\alpha$ and $\alpha'$ are each
multiplied by $\lambda$ according to the previous formulas (\ref{alphaDel}) and
(\ref{alphaprimeDel}).

We consider the previous situation where a miner repeats some block withholding strategy
(deviant or not) and the rest of the network is
composed by honest miners. We assume that the total hash power remains constant. 
The first observation is that this miner can have a significant impact on the first difficulty adjustment if he conceals his 
blocks. For
example, after a period of mining $n_0$ official blocks, if a miner adopts a
block withholding strategy, this slows down the natural progression of the
official blockchain. The network will generally take longer than the expected
two weeks to mine $2016$ blocks \cite{GPM2023}. This triggers a downward difficulty adjustment, making the
mining activity more profitable afterwards. If in the long run the miner continues with the same
strategy, the difficulty parameter will stabilize assuming that the total hashrate remains constant. Then,
in such a period of constant difficulty, since the interblock time is $10$
minutes, the duration $\tau$ of an attack cycle is proportional on
average to the height progression $H$ of the official blockchain during that
cycle, i.e.  $\mathbb{E} [\tau] =\mathbb{E} [H] \times 10$. Although it is intuitively 
clear, this is proved by direct application of Wald's
Theorem.

Thus we have the following Proposition \cite{GPM2018}

\begin{proposition}
  \label{gamma}Consider a finite mining strategy, i.e. $\mathbb{E} [\tau] < \infty$. 
  Let $G$ be the number of blocks
  per cycle mined by the miner added to the official blockchain and $H$
  denote the height progression of the official blockchain over a cycle. Then
  the profitability ratio of the strategy is $\Gamma = \frac{\mathbb{E}
  [G]}{\mathbb{E} [H]}$.
\end{proposition}

In other words, among blocks mined by miners, only blocks that will eventually become part of the official blockchain give rise to rewards. The objective function (over the long term) is the proportion of blocks mined by the attacker present in the official blockchain which is exactly $\frac{R_1+R_2+..}{H_1+H_2+...}=\Gamma$.

Note that the honest strategy corresponds to a simple cycle that ends each time that a 
block is discovered. In this case, we have $\mathbb{E} [G] = q$ and
$\mathbb{E} [H] = 1$ where $q$ is the relative hash power of the miner. Thus,
$\Gamma = q$ for the honest strategy and a rational miner has an incentive to
adopt a deviant strategy if and only if $\Gamma > q$.

In the next section, we show that there are mining strategies with $\Gamma$ greater than 
$q$ (Section \ref{bwa}) and that a modification of the difficulty adjustment formula in the Bitcoin 
protocol results in $\Gamma$ always less than $q$ (Section \ref{mbp}).

\section{Block withholding attacks}\label{bwa}

It has been known since late 2013 that the rules of the Bitcoin protocol are
not aligned with the interests of miners. 
The reader can consult  \cite{GPM2020} or \cite{Zhang} for an overview.
In particular, when the
miner has enough computing power, he can have an interest in adopting a deviant strategy. 
For the sake of being self-contained, we present next a simple example of deviant strategy, the ``1+2 strategy'', 
showing that the honest strategy is not the most profitable one. 
This strategy is a simplified version of the selfish-mining strategy \cite{ES14}. See Note \ref{ssm} below. 
It is the simplest profitable block withholding mining strategy one can imagine.

\begin{note}\label{objective}
The profitability ratio $\Gamma$ that we have introduced equals to the objective function defined in some literature 
as $\frac{r_1+r_2+...}{h_1+h_2+...}$ where each $r_i$ (resp. $h_i$) denotes the reward in terms of number of blocks won by the attacker 
(resp. progression of the height of the official blockchain) coming from a new step in the random walk \cite{BET20, SSZ17}. 
To prove that the two notions are the same, we can group terms in the numerator which belong to a same excursion around state $\{0\}$ : 
$$
r_1+r_2+..=(r_1+...+r_{\nu_1})+(r_{\nu_1+1}+...+r_{\nu_2})+...
$$ 
where $\nu_i$ is the $i-th$ returning time at the initial state $\{0\}$. The same thing occurs for the denominator. So, 
$$
\frac{r_1+r_2+..}{h_1+h_2+...}=\frac{R_1+R_2+..}{H_1+H_2+...}=\Gamma
$$ 
\end{note}

\medskip

\textbf{The ``1+2 strategy''.}

\medskip

Suppose that the network is composed of an attacker and a set of
honest miners. The attack is
defined as follows. The attacker starts mining. If the honest miners are first 
to find a block then the attack ends and the attacker returns 
to mine on top of the last discovered
block. If, on the other hand, the attacker manages to mine a block
before the honest miners, then he keeps it
secret and continues to mine on top of that block,
seeking to widen the gap with the official blockchain. Then, regardless of the
identity of the miners who validated the following blocks, as soon as two
blocks have been discovered, the attack ends. If he is successful, i.e. if it has mined
more blocks than the honest miners, the attacker reveals his secret
blocks and imposes its "fork" (his block sequence) on the official blockchain,
which then goes through a small reorganization. If this is not the case, it is
useless for the attacker to broadcast anything since none of his blocks will be 
included in the official blockchain.

If we denote by A a block discovered by the attacker and by B a block
discovered by the honest miners, the outcome of the attack can be encoded as a word formed with
the letters A and B. The universe $\Omega$ of all possible cycles is :
\[ \Omega = \{ B, AAA, AAB, ABA, ABB \} \]
The name ``1+2 strategy'' reflects the fact that the
attacker is waiting to discover a block (hence the ``1''). 
Then, when he does, he waits
for two blocks to be discovered, hence the "+2".

\begin{example}
  We describe the cycle ABA. The attacker discovers a block, then the honest 
  miners mine one too and finally the attacker 
  mines the next one.  In this case,
  the attack has succeeded and the attacker propagates his two secret blocks (the
  two A's). The official blockchain goes through a small reorganization: its
  now penultimate block (block "B") is replaced (it becomes an orphan block) and
  the height of the blockchain has increased by one block. 
  The attacker earns the reward contained from his two
  blocks.
\end{example}

We can now compute the expected profitability of the ``1+2'' strategy.
As before, we denote by $q$ the relative hash power of the attacker and $p = 1 - q$ the
one of the honest miners. Let $G$ the number of
blocks mined by the attacker and added to the official blockchain.
Let also $H$ be the progression of the height of the official blockchain during an attack cycle. 
We have the probability distribution
\[ \mathbb{P} [B] = p, \mathbb{P} [AAA] = q^3, \mathbb{P} [AAB]
   =\mathbb{P} [ABA] = pq^2, \mathbb{P} [ABB] = p^2 q \]
From this we can compute:
\[ G (B) = G (ABB) = 0, G (AAA) = 3, G (AAB) = G
   (ABA) = 2 \]
and
\[ H (B) = 1, H (ABB) = H (AAB) = H (ABA) = 2, H
   (AAA) = 3 \]
Therefore we can compute the expected values
\begin{equation*}
  \mathbb{E} [G] = p \cdot 0 + q^3 {\cdot} 3 + pq^2  \cdot 2 + pq^2 
  \cdot 2 + p^2 q \cdot 0 = q^2(4-q) \label{eg} 
\end{equation*}
and
\begin{equation*}
\mathbb{E} [H]  =  p \cdot 1 + q^3 {\cdot} 3 + pq^2  {\cdot} 2 + pq^2 
\cdot 2 + p^2 q \cdot 2 = 1+q+q^3  \label{eh}
\end{equation*}

From these computations it follows:
\begin{proposition}\label{12}
  The profitability ratio of the "1+2" strategy is 
  $$
  \Gamma=\frac{\mathbb{E}
  [G]}{\mathbb{E} [H]} =\frac{q^2 \cdot (4 - q)}{1 + q
  + q^3}
  $$
\end{proposition}

This proves that the "1+2 strategy`` is more profitable than the honest strategy if
and only if $\frac{q^2 \cdot (4 - q)}{1 + q + q^3} > q$. \ This 
happens only when
\begin{equation}
  q > \sqrt{2} - 1 \label{seuil2}
\end{equation}
Thus, if a miner has a has a hashpower of more than $\sqrt{2} - 1 \approx 41, 4$\%,
then he has no incentive to follow the protocol. From this simple example we reach the conclusion:
\begin{proposition} \label{prop:align}
  The rules of the Bitcoin protocol are not aligned with the self-interests of
  miners.
\end{proposition}

\begin{note}\label{ssm}
The strategy can also be described with the help of the following finite Markov chain with only four states.
\begin{figure}[htb]
\centering
\begin{tikzpicture}
\node[state] (s1) {0'};
\node[state, below right of=s1] (s2) {1};
\node[state, below left of=s1] (s3) {0};
\node[state, right of=s2] (s4) {2};

\draw (s1) edge[bend right] node {$p$} (s3);
\draw (s1) edge[bend left] node {$q$} (s3);

\draw (s2) edge node {$p$} (s1);
\draw (s2) edge node {$q$} (s4);

\draw (s3) edge node {$q$} (s2);
\draw (s3) edge[loop left] node {$p$} (s3);

\draw (s4) edge[bend left=20] node {$p$} (s3);
\draw (s4) edge[bend left=30] node {$q$} (s3);

\end{tikzpicture}
\end{figure}
This is a truncated version of the "selfish mining" random walk where states $\{n\}$ with $n>2$ have disapeared \cite{ES14}. State $\{0\}$ is the state when all miners mine on top of the same last block of the official blockchain. State $\{k\}$ for $k=1,2$ is the state when the attacker has just validated $k$ secret block(s) and is mining on top of his secret fork whereas the honest miners haven't found any block yet.
State $\{0'\}$ is the state when both the attacker and the honest miners have mined a block and are trying to mine one another on top of their last block. Each state transition gives a reward $r$ to the attacker and a contribution $h$ to the increase of the height of the official blockchain. A possible modelization for $(r,h)$ is the following. We have $(r,h)=(0,1)$ from $\{0\}$ to $\{0\}$, $(r,h)=(0,0)$ from $\{0\}$ to $\{1\}$, $(r,h)=(0,1)$ from $\{1\}$ to $\{0'\}$, $(r,h)=(0,1)$ from $\{0'\}$ to $\{0\}$ if the
honest miners find the new block (probability $p$) and $(r,h)=(2,1)$ otherwise (probability $q$). To simplify, the connectivity $\gamma$ is set $\gamma =0$. So, there are only two edges from $\{0'\}$ to $\{0\}$ depending on who mine the next block between the honest miners (probability $p$) and the attacker (probability $q$). We also have $(r,h)=(0,0)$ from $\{1\}$ to $\{2\}$ and $(r,h)=(2,2)$ from $\{2\}$ to $\{0 \}$ if the
honest miners find the new block (probability $p$) and $(r,h)=(3,3)$ otherwise (probability $q$). This is just one possible modelization. We could set for instance $(r,h)=(2,2)$ from $\{1\}$ to $\{2\}$ and $(r,h)=(0,0)$ from $\{2\}$ to $\{0 \}$ if the
honest miners find the new block (probability $p$) and $(r,h)=(1,1)$ otherwise (probability $q$). The contribution $r$ represents a number of blocks that will end soon or later in the official blockchain (this happens when the walker comes back to state 0). The only quantities of interest are $\sum_{i=1}^{\nu} r_i$ and $\sum_{i=1}^{\nu} h_i$ 
 where $\nu$ is the first returning time at $0$ and $r_i$ is the reward won by the attacker coming from the transition from state $\{i-1\}$ to state $\{i\}$. Similarly, $h_i$ is the contribution to the increase of the height of the official blockchain following a change of state from state $\{i-1\}$ to state $\{i\}$. The finite Markov chain has a stationary probability distribution $\pi$ which is given by $\pi(\{0 \})=\frac{1}{1+2q}$, $\pi(\{k \})=q.\pi(\{k-1 \})$ for $k>0$ and $\pi(\{0' \}=\frac{p q}{1+2q}$. From here, we find easily that $\EE[r]=\frac{q^2}{1+2 q}(4-q)$ and 
$\EE[h]=\frac{1+q+q^3}{1+2 q}$. Hence we get again $\Gamma=\frac{\EE[r]}{\EE[h]}=\frac{q^2(4-q)}{1+q+q^3}$ as in Proposition \ref{12}.
\end{note}

By enumerating the other possible strategies, one can show that the ``1+2'' strategy is the best possible when the miner’s attack cycle ends after the discovery of three blocks.

\section{Modified Bitcoin protocol}\label{mbp}

From Corollary \ref{cor:optimal} and Proposition \ref{prop:align} it appears that the origin of the problem comes from the difficulty adjustment formula that can be exploited by blockwithholding strategies. We can prevent these attacks using a different formula for the difficulty adjustment taking into account orphan blocks.

\subsection{A more general difficulty adjustment formula}

We consider a more general difficulty adjustment formula on the Bitcoin network of the form :
\begin{equation}
  \Delta' = \Delta \cdot \frac{D \times 10}{T}  \label{cible2}
\end{equation}
where $D$ denotes the progression of a certain quantity, called \textit{difficulty
function}, over a validation period of $n_0$ official blocks. 
The introduction of a difficulty function and the notation $D$ comes from \cite[p. 116]{BET20}. The authors use the term ``difficulty contribution''.
In the case of current Bitcoin protocol we simply have that $D$ grows linearly and at the end of the 
period $D = n_0 = 2016$. That is, 
the difficulty function for Bitcoin increases one by one with each
new block on the official blockchain. 

We can consider more general difficulty functions $D$. With the new formula, 
the mining time of a cycle is now
proportional to the progression of the difficulty function. The same argument
as before shows the following result 
(see also \cite{BET20}).

\begin{proposition}
  In the context of a modified Bitcoin protocol with a difficulty adjustment mechanism
  following relation (\ref{cible2}), the rate of return of a mining strategy
  is $\Gamma = \frac{\mathbb{E} [G]}{\mathbb{E} [D]}$ where $G$ is the number
  of blocks per cycle mined by the miner and added to the official blockchain
  and $D$ denotes the progression of the difficulty function over that cycle.
\end{proposition}

Then we can consider a modification of the Bitcoin protocol, where miners, 
in addition to their mining activity, will report the presence of orphan
blocks by recording their proof of existence. It is
possible to encourage miners to record the existence of orphan blocks by
modifying the rule that defines the official blockchain and, as we will see later, 
even by rewarding the reporting of orphan blocks. The blockchain that maximizes the
difficulty is the one reporting the most number of orphan blocks weighted by their difficulty 
(during a difficulty 
adjustment period it will be the longest). At
the end of a validation period of $n_0$ official blocks, the adjustment of difficulty 
will be given by a similar formula as the standard one but of the form:
\begin{equation}
  \Delta'  =  \Delta \cdot \frac{(n_0 + n_1) \times 10}{T}  \label{n0n1}
\end{equation}
where $n_1$ is the number of orphaned blocks reported during the last
mining period of $2016$ blocks. That means that we consider a difficulty
function $D$ that is not given by the height function of the blockchain
but that increases by $1$ at each time that a new block is
registered in the official blockchain (whether this is an official block
or just an orphan block detected by the network). The honest mining strategy 
corresponds to a cycle that ends as soon as a block is discovered, which still
gives a profitability ratio equal to $q$.
Hence, we have the following Proposition:

\begin{corollary}
  Consider a finite mining strategy with a length cycle $\tau$ and $\mathbb{E} [\tau] <
  \infty$. The number of blocks added to the official blockchain by the
  miner (resp. the progression of the difficulty function) between $0$ and
  $\tau$ is $G (\tau)$ (resp. $D (\tau)$). This strategy is more profitable
  than the honest strategy if and only if 
  $$
  \mathbb{E} [G (\tau)] > q\mathbb{E}
  [D (\tau)]
  $$
\end{corollary}

\subsection{Stability of the Nakamoto consensus with general difficulty adjustment formula}

Now assume that we are running the modified Bitcoin protocol with a difficulty adjustment formula as 
in (\ref{n0n1}) that takes into account orphan blocks.

\begin{theorem}\label{opt00}
  Let's assume that sooner or later orphaned blocks mined by honest miners will be detected, reported and recorded in the blockchain (nothing is assumed a priori about orphaned blocks mined by the attacker). Then, for any finite mining strategy  with $\mathbb{E} [\tau] < \infty$,
  we have
  \begin{equation*}  
    \mathbb{E} [G (\tau)] \leqslant q\mathbb{E} [D (\tau)]
  \end{equation*}

\end{theorem}

\begin{proof}
  During an attack cycle, among the $N' (\tau)$ blocks that are mined by
  the attacker we consider  $\Orph_A$, resp. $\Off_A$, the number of official, resp. orphan, blocks. Also, we denote $N (\tau)$ blocks
  mined by honest miners and among them $\Off_H$, resp. $\Orph_H$, the
  numbers of official resp. orphan, blocks mined by honest miners. We have
  \begin{align*}
    N (\tau) &= \Off_H + \Orph_H\\
    N' (\tau) &= \Off_A + \Orph_A
  \end{align*}
  and $G (\tau) = \Off_A$.
  
  Orphan blocks by honest miners are public and will be registered
  sooner or later in the official blockchain. Only the orphan blocks of the
  attacker can remain secret. Therefore we have,
  \[ \Off_A + \Off_H + \Orph_H \leqslant D (\tau) \]
  The two processes $N$ and $N'$ are Poisson processes of parameters $\lambda
  \cdot p$ and \ $\lambda \cdot q$ where $\lambda$ depends on the
  difficulty adjustment. If all miners are
  honest, then $\lambda = \frac{1}{\tau_0}$ with $\tau_0 = 10$ minutes. The condition 
  $\mathbb{E} [\tau] < \infty$ gives that $\mathbb{E} [N (\tau)] =
  \lambda p\mathbb{E} [\tau]$ and $\mathbb{E} [N' (\tau)] = \lambda
  q\mathbb{E} [\tau]$. This follows from the fact that if $M$ is a Poisson
  process of parameter $\alpha$ then the compensated process $M (t) - \alpha
  t$ is a martingale. Therefore,
  \[ p\mathbb{E} [\Off_A] \leqslant p\mathbb{E} [N' (\tau)] = p \lambda
     q\mathbb{E} [\tau] = q \lambda p\mathbb{E} [\tau] = q\mathbb{E} [N
     (\tau)] = q\mathbb{E} [\Off_H] + q\mathbb{E} [\Orph_H] \]
  which gives,
  \begin{align*}
    \mathbb{E} [G (\tau)] &= \mathbb{E} [\Off_A]\\
    &= p\mathbb{E} [\Off_A] + q\mathbb{E} [\Off_A]\\
    & \leqslant q\mathbb{E} [\Off_H] + q\mathbb{E} [\Orph_H] +
    q\mathbb{E} [\Off_A]\\
    & \leqslant q \cdot \mathbb{E} [D (\tau)]
  \end{align*}
  which proves the Theorem.
\end{proof}

\begin{corollary}\label{opt01}
  In the modified Bitcoin protocol, the most profitable strategy is always the honest one, and this 
  does not depend on the connectivity of the miner.
\end{corollary}

\begin{note}
Orphan blocks mined by the attacker are not necessarily recorded by the official blockchain. However, it is essential that all orphan blocks mined by honest miners are recorded. If this is not the case, the theorem may be invalidated. For example, in \cite{SZNP2024}, the authors consider difficulty adjustment proposals that sometimes allow avoiding the reporting of orphan blocks. This is particularly the case in their proposals if an orphan block is mined in a different period than the official block that reports it. Under these conditions, Theorem \ref{opt00} does not hold.
\end{note}

We can be more precise. We have indeed:
\begin{eqnarray}
  p\mathbb{E} [\Off_A] + p\mathbb{E} [\Orph_A] & = & p\mathbb{E}
  [N' (\tau)] \nonumber\\
  & = & p \lambda q\mathbb{E} [\tau] \nonumber\\
  & = & q \lambda p\mathbb{E} [\tau] \nonumber\\
  & = & q\mathbb{E} [N (\tau)] \nonumber\\
  & = & q\mathbb{E} [\Off_H] + q\mathbb{E} [\Orph_H]  \label{pre}
\end{eqnarray}
Among the $\Orph_A$ orphan blocks of the attacker, we note
$\Orph'_A$ the orphan blocks made public by the attacker and thus
detected by the network at a given moment (the others remain secret).

We have, on the one hand
\[ \mathbb{E} [D (\tau)] =\mathbb{E} [\Off_H] +\mathbb{E}
   [\Orph_H] +\mathbb{E} [\Off_A] +\mathbb{E} [\Orph'_A] \]
and on the other hand according to (\ref{pre}),
\begin{eqnarray*}
  \mathbb{E} [\Off_A] +\mathbb{E} [\Orph'_A] & = & p\mathbb{E}
  [\Off_A] + p\mathbb{E} [\Orph'_A] + q\mathbb{E} [\Off_A] +
  q\mathbb{E} [\Orph'_A]\\
  & \leqslant & p\mathbb{E} [\Off_A] + p\mathbb{E} [\Orph_A] +
  q\mathbb{E} [\Off_A] + q\mathbb{E} [\Orph'_A]\\
  & \leqslant & q\mathbb{E} [\Off_H] + q\mathbb{E} [\Orph_H] +
  q\mathbb{E} [\Off_A] + q\mathbb{E} [\Orph'_A]\\
  & \leqslant & q \cdot \mathbb{E} [D (\tau)]
\end{eqnarray*}
Now assume that the protocol grants a reward $x$ to any orphan block creator
with $x \leq 1$. Then, \
\begin{eqnarray*}
  G (\tau) & = & \Off_A + x \cdot \Orph'_A\\
  & \leqslant & \Off_A + \Orph'_A
\end{eqnarray*}
So, from the above,
\begin{eqnarray*}
  \mathbb{E} [G (\tau)] & \leqslant & \mathbb{E} [\Off_A] +\mathbb{E}
  [\Orph'_A]\\
  & \leqslant & q \cdot \mathbb{E} [D (\tau)]
\end{eqnarray*}
Thus we have proved the more general result:

\begin{theorem}\label{2nd}
  Consider a modified Bitcoin protocol that grants a coinbase fraction reward $0\leq x\leq 1$ for each block 
  to orphan block creators. We assume that the difficulty adjustment mechanism is given by (\ref{n0n1}). 
  Then the honest mining strategy is optimal. When $x < 1$ this is the only optimal strategy.
\end{theorem}

\begin{remark}
The result of Theorem \ref{2nd} may seem counter-intuitive. Granting a reward to the creators of orphaned blocks seems like an incentive to carry out block-withholding attacks, since even when the attack cycle fails, the dishonest miner will be partially rewarded for creating his orphaned blocks. At the same time, however, the inclusion of these orphan blocks in the new difficulty adjustment algorithm means that the mining difficulty parameter is no longer artificially lowered. Per attack cycle, the average revenue $\EE[R]$ increases, as does the difficulty function $\EE[D]$, and the calculation of the profitability ratio $\frac{\EE[R]}{\EE[D]}$ shows that it is lower than its value in the case where the miner mines honestly. Note also that the result of Theorem \ref{2nd} is purely theoretical and in no way encourages to reward the creators of orphan blocks. As we have seen in Theorem \ref{opt00} and Corollary \ref{opt01}, the result is true even if no reward is given ($x=0$). The main shortcoming of the idea of remunerating orphan block creators is that it creates monetary inflation, which changes the 21 million limit if nothing else is modified. However, we can also adjust the rule for the timing of the halving by triggering it when the total emission equals that of a classical halving period. In this case, the 21 million limit of bitcoins is preserved.
\end{remark}

To illustrate this theorem, we can revisit the example of the ``1+2" strategy.
\begin{example}
In the case of the ``1+2" strategy, we have 
\[ D (B) = 1, D (ABB) = 2, D (AAA) = D (AAB) = D(ABA) = 3 \] and we check that 
$\mathbb{E} [G] -q \mathbb{E} [D]  = -p^3 q<0$. Note that if the attacker modifies his strategy and decides to publish his orphan block (the ``A" in the sequence $ABB$), then in this case we have:
\[ G (B) = 0, G (ABB) = x, G (AAA) = 3, G (AAB) = G
   (ABA) = 2 \]
and
\[ D (B) = 1, D (ABB) = D (AAA) = D (AAB) = D(ABA) = 3 \]
Similarly, we check that
$\mathbb{E} [G] -q \mathbb{E} [D]  = -p^2 q (1-x)\leq 0$.
\end{example}

\section{Conclusion}
Regularly, miners add blocks to a rooted graph.
Block counting processes are Poisson processes.
The official blockchain is in practice the longest sequence of blocks that can be found there starting from the genesis block.
Only blocks that are part of the official blockchain give rise to rewards.
In the absence of a Difficulty Adjustment Algorithm (DAA), the best mining strategy is the honest strategy which consists in always mining on the last block of the official blockchain and always immediately revealing block discoveries. This is no longer the case with the current DAA implemented in Bitcoin today. In this case and in the long term, the profitability ratio of a mining strategy is the proportion of mined blocks present in the official blockchain. In some cases, block withholding strategies outperform the honest strategy. This is sometimes the case with the 1+2 strategy that we have described here and which is extremely simple. However, it suffices to slightly modify the DAA by taking into account the production of orphan blocks to make these attacks ineffective because they are not profitable.
It is even possible to partially remunerate the creators of orphan blocks without jeopardizing the security of the protocol.
The two general results presented here (Theorem \ref{1st} and Theorem \ref{2nd}) are direct applications of Doob's stopping time theorem. However, proving these theorems rigorously seems impossible without the use of martingale theory.

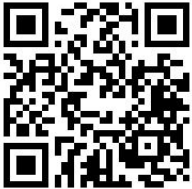
\end{document}